\documentclass[12pt,a4paper]{article}
\usepackage[left=1in,right=1in,top=1in,bottom=1in]{geometry}
\usepackage[utf8]{inputenc} % allow utf-8 input
\usepackage[T1]{fontenc}	% use 8-bit T1 fonts
\usepackage{libertine}
\usepackage{libertinust1math}
\usepackage[english]{babel}
\usepackage{amsmath}
\usepackage{amsfonts}
\usepackage{amssymb}
\usepackage{amsthm}
\usepackage{mathtools}
\usepackage{thm-restate}
\usepackage[linesnumbered,vlined,ruled]{algorithm2e}

\usepackage{tikz}
\usepackage{xcolor}
\usepackage[inline]{enumitem}
\usepackage{hyperref}
\usepackage[justification=centering]{caption}
\usepackage{subcaption}
\usetikzlibrary{calc,arrows,automata,graphs,shapes,petri,decorations.pathmorphing,decorations.markings}
\tikzset{->-/.style={decoration={markings,mark=at position .5 with {\arrow{>}}},postaction={decorate}}}
\let\coloneqq\coloneq
\let\emptyset\varnothing
\let\epsilon\varepsilon
\let\rho\varrho
\let\theta\vartheta
\let\phi\varphi

\DeclareMathOperator{\contract}{/}
\DeclareMathOperator{\delete}{\backslash}
\DeclareMathOperator{\rank}{\rho}

\definecolor{rwthblue}{cmyk}{1, .5, 0, 0}
\definecolor{rwthlightblue}{cmyk}{.45, .14, 0, 0}
\definecolor{rwthred}{cmyk}{.15, 1, 1, 0}
\definecolor{rwthgreen}{cmyk}{.7, 0, 1, 0}
\definecolor{rwthorange}{cmyk}{0, .4, 1, 0}
\definecolor{rwthmagenta}{cmyk}{0, 1, .25, 0}
\definecolor{rwthyellow}{cmyk}{0, 0, 1, 0}
\definecolor{rwthpetrol}{cmyk}{1, .3, .5, .3}
\definecolor{rwthturquoise}{cmyk}{1, 0, .4, 0}
\definecolor{rwthmay}{cmyk}{.35, 0, 1, 0}
\definecolor{rwthbordeaux}{cmyk}{.25, 1, .7, .2}
\definecolor{rwthviolet}{cmyk}{.7, 1, .35, .15}
\definecolor{rwthpurple}{cmyk}{.6, .6, 0, 0}

\tikzset{vertex/.style={draw,circle,inner sep=0pt,minimum size=5pt}}
\tikzset{labeled vertex/.style={draw,circle,inner sep=0pt,minimum size=15pt}}
\tikzset{hidden/.style={draw=none,circle,inner sep=0pt,minimum size=5pt}}
\tikzset{label/.style={font=\scriptsize}}
\tikzset{player1/.style={black,solid}}
\tikzset{player2/.style={rwthred,dashed}}
\tikzset{player3/.style={rwthblue,dotted,thick}}
\tikzset{player4/.style={rwthgreen,dashdotted}}
\tikzset{monopsony/.style={#1,very thick,rounded corners}}

\newcommand{\I}{\ensuremath{\mathcal{I}}}
\newcommand{\B}{\ensuremath{\mathcal{B}}}
\newcommand{\C}{\ensuremath{\mathcal{C}}}
\newcommand{\Z}{\ensuremath{\mathbb{Z}}}

\newcommand{\hp}{\ensuremath{\hat{p}}}
\newcommand{\hB}{\ensuremath{\hat{B}}}
\newcommand{\Mnat}{\ensuremath{\mathrm{M}^\natural}}

\theoremstyle{definition}
\newtheorem{definition}{Definition}
\newtheorem*{example*}{Example}

\theoremstyle{plain}

\newtheorem{theorem}[definition]{Theorem}

\title{A Simplified Analysis of the Ascending Auction to Sell a Matroid Base}
\date{November 25, 2024}
\author{
	Britta Peis\thanks{School of Business and Economics, RWTH Aachen University}\\
	britta.peis@oms.rwth-aachen.de
	\and
	Niklas Rieken\footnotemark[1]\\
	niklas.rieken@oms.rwth-aachen.de
}

\begin{document}
\maketitle
\begin{abstract}
	We give a simpler analysis of the ascending auction of Bikhchandani, de Vries, Schummer, and Vohra \cite{bikhchandani2011ascending} to sell a welfare-maximizing base of a matroid at Vickrey prices.
	The new proofs for economic efficiency and the charge of Vickrey prices only require a few matroid folklore theorems, therefore shortening the analysis of the design goals of the auction significantly.
\end{abstract}

\section{Introduction}\label{sec:intro}
We consider a market consisting of a set of $m$ indivisible distinguishable items $E$ and a set of $n$ buyers $N$, where each buyer $i \in N$ is interested in a subset $E_i
\subseteq E$.
Each buyer $i \in N$ has a valuation $v_i(e) \in \Z_+$ for every $e \in E_i$ which may be interpreted as buyer $i$'s maximum willingness to pay for item $e$.
We may also write $v_i(X) = \sum_{e \in X} v_i(e)$ for any subset $X \subseteq E_i$ to denote the additive valuation of any buyer.
We consider the particular case where the auctioneer is constrained to sell a base of a matroid (see Section~\ref{sec:matroids} for basics on matroid theory).

This model has been studied by Bikhchandani, de Vries, Schummer, and Vohra in the paper ``An Ascending Vickrey Auction for Selling Bases of Matroids'' presented at SODA in 2008 \cite{bikhchandani2008ascending} and published in the journal OR in 2011 \cite{bikhchandani2011ascending}.
In the sequel, we only reference the journal version.
Bikhchandani et al.\ present an ascending auction for selling a base of a matroid, and prove that the auction fulfills the following three important mechanism design goals:
\begin{enumerate*}[label=(\roman*)]
	\item signaling truthfully is an ex-post equilibrium,\label{goals:eq}
	\item the truthful equilibrium results in an efficient (i.e.\ welfare-maximizing) outcome, and \label{goals:opt}
	\item the auction runs in strongly polynomial time.\label{goals:poly}
\end{enumerate*}
The main building block for showing that the ascending auction (described in Section~\ref{sec:auction} along with some basics on auction theory) fulfills properties \ref{goals:eq}-\ref{goals:poly} is Theorem~13 in \cite{bikhchandani2011ascending} (or Theorem~\ref{thm:main-theorem} below).
It states that the auction computes in strongly polynomial time a welfare-maximizing base (which yields \ref{goals:opt} and \ref{goals:poly} immediately) and charges Vickrey prices (which is the key for proving \ref{goals:eq}).
This paper provides a simpler proof of the theorem using only some matroid folklore theorems.

While at first sight the restriction to auctions in which the auctioneer is constrained to sell a base of a matroid seems to be an abstract model in the mechanism design literature, it turns out to contain various interesting special cases.
As described in \cite{bikhchandani2011ascending}, examples include scheduling matroids (cf. \cite{demange1986multi}), the allocation of homogeneous objects (cf. \cite{ausubel2004efficient}), pairwise kidney exchange (cf. \cite{roth2005pairwise}), spatially distributed markets (cf. \cite{babaioff2004mechanisms}), bandwidth markets (cf. \cite{tse1998multiaccess}), and multi-class queueing systems (cf. \cite{shanthikumar1992multiclass}).
Some of the examples require a generalization of the model to the setting where the auctioneer is constrained to sell a base of an integral polymatroid, and the buyers $i \in N$ have non-decreasing, concave, and piece-wise-linear valuations with integral breakpoints.
However, as also observed in \cite{bikhchandani2011ascending}, this more general model can easily be reduced to our setting via Helgason's pseudo-polynomial reduction (cf. \cite{helgason1974aspects}).
More recently, an extension of this auction that sells a base of a polymatroid which runs in polynomial time has been published in Raach's PhD thesis \cite{raach2023implementing}.
Its analysis is still grounded on the original work in \cite{bikhchandani2011ascending}.
In particular, the analysis involves Theorem~\ref{thm:main-theorem} to show that the updated auction computes VCG prices and its truthfulness also follows from the original paper.

It is well-known that there exists a \emph{sealed-bid} auction for our model which fulfills the three mechanism design goals \ref{goals:eq}-\ref{goals:poly}, namely the VCG mechanism \cite{vickrey1961counterspeculation,clarke1971multipart,groves1973incentives}.
The VCG mechanism operates by first asking the buyers for their valuations for the items.
Based on the reported valuations, which may or may not coincide with the true valuations, the VCG mechanism computes a welfare maximizing base, and charges Vickrey prices to every buyer.
Due to the fact that the auctioneer is constrained to sell a base of a matroid, the computation of the welfare maximizing base and the Vickrey prices can be done in strongly polynomial time with the matroid greedy algorithm.
There are, however, various reasons for preferring an ascending auction over a sealed-bid auction.
 
In general, in an ascending auction, the auctioneer announces prices, and the buyers report their demands at announced prices.
The auction ends once the demands are feasibly satisfied.
Otherwise, the auctioneer increases the prices.
The here presented auction follows this outline and we describe it more precisely in Section~\ref{sec:auction}.
Typically ascending auctions have the drawback that they are slower than sealed-bid auctions (since they require multiple rounds of communication as bid elicitation process).
However, in many settings ascending auctions are still preferred over sealed-bid auctions as they are more transparent, require less information to be revealed by the buyers, and less bits of communication.
We briefly elaborate on these aspects in Appendix~\ref{apx:tpc}.
For an extended discussion of advantages of ascending auctions over sealed-bid auctions see \cite{ausubel2004efficient} and \cite{cramton1998ascending}.

\subsection{Related Literature}
Since there is an extensive literature on ascending auctions, we restrict in this paper only to those references that are closest to our model and our results, and refer to the references in e.g. Chapter~11 in \cite{AGTbook} or \cite{parkes2006iterative,blumrosen2005computational} for a broader overview on the theory of ascending auctions.

Matroids play an important role in mechanism design theory, and particularly in the analysis of ascending auctions.
To mention a prominent example, consider the general setting of auctions where buyers are interested in buying subsets (\emph{bundles}) of items and have private valuations over the bundles.
Walrasian equilibrium prices (cf. \cite{walras1874}) are guaranteed to exist if the valuation functions satisfy the \emph{gross substitutes} condition \cite{kelso1982job,gul1999walrasian}.
In fact, by the Maximal Domain Theorem of Gul and Stacchetti \cite{gul1999walrasian}, gross substitute valuations are the largest class containing unit-demand valuations that guarantee existence of a Walrasian equilibrium.
The gross substitutes condition is closely related to matroids in the sense that the gross substitutes condition is equivalent to \Mnat-concavity \cite{fujishige2003note,murota2013computing}, implying that the inclusion-wise minimal preferred bundles form the base set of a matroid for any given price vector (see e.g. \cite{gul2000english}).
The rich theory of matroids, polymatroids, and discrete convexity turned out to be extremely helpful for understanding the computational complexity of iterative auctions with gross substitutes valuations (see e.g. the survey \cite{leme2017gross} and \cite{murota2013computing,murota2016time,eickhoff2023faster}).

The preferences in our model also satisfy the gross substitutes condition, so that our model fits into the frameworks analyzed in \cite{ausubel2006efficient,parkes2000iterative,ausubel2002ascending,devries2007ascending,mishra2007ascending}.
In all of these frameworks, efficient outcomes in equilibrium are achieved.
However, these more general frameworks rely on an exponential (in the number of items) number of the price-raising steps.
In contrast, the ascending auction for our model runs in strongly polynomial time without the necessity to run several dummy auctions (like in \cite{ausubel2006efficient}), and without relying on a proxy bidding scheme (like in \cite{ausubel2002ascending} and \cite{parkes2000preventing}).
See \cite{bikhchandani2011ascending} for a more detailed comparison with the above mentioned frameworks and also \cite{raach2023implementing} for the extension to the polymatroid case.

\subsection{Contribution and Overview}
We provide a simplified proof for showing that the ascending auction described in \cite{bikhchandani2011ascending} returns a welfare-maximizing base and charges Vickrey prices.
The original proof requires a very microscopic analysis of the auction:
In every step, one has to keep track of deleted items and cocircuits that act as a witness for price-setting items, building an auction history, which is called the \emph{VCG sequence}.
In contrast, our proof does not require this step-by-step view of the auction and instead only uses three matroid folklore lemmas which reduces the length of the proof significantly, is less technical and less index-heavy.
The paper as a whole, however, is still self-contained: We introduce all terminology needed from matroid theory in Section~\ref{sec:matroids}, give a slightly amended description of the auction with examples in Section~\ref{sec:auction}, and the aforementioned simplified analysis in Section~\ref{sec:analysis}.

\section{Basics on Matroid Theory}\label{sec:matroids}
We consider a finite ground set $E$, which in the sequel will be the set of items.
For the sake of convenience, we may write $X+e \coloneqq X \cup \{e\}$, and $X-e \coloneqq X \setminus \{e\}$, for any subset $X \subseteq E$ and any element $e \in E$.
Moreover, given a function $w\colon E \to \mathbb{R}$, we may write $w(X) \coloneqq \sum_{e \in X} w(e)$.

This section contains some basic definitions and facts from the theory of matroids, which originated by the work of Whitney \cite{whitney1935abstract}.
A modern textbook on matroid theory is the book by Oxley~\cite{oxley2006matroid}.
Our notation is mainly based on Chapter~39 in the book by Schrijver~\cite{schrijver2003combinatorial}.
Readers familiar with matroid theory are encouraged to skip this section.

We say a pair $M = (E, \I)$ consisting of a finite \emph{ground set} $E$ and a family of \emph{independent sets} $\I \subseteq 2^E$ is a \emph{matroid} if 
\begin{enumerate*}[label=(I\arabic*)]
	\item $\emptyset \in \I$,\label{M:existence}
	\item for all $I_1 \in \I$ and $I_2 \subseteq I_1$, it holds that $I_2 \in \I$, and\label{M:hereditary}
	\item for all $I_1, I_2 \in \I$ with $|I_1| > |I_2|$, there exists $e \in I_1 \setminus I_2$, such that $I_2+e \in \I$.\label{M:augmentation}
\end{enumerate*}

An inclusion-wise maximal independent subset of $E$ is called a \emph{base} and we denote the set of bases of a matroid by $\B$.
For example, if $G = (V, E)$ is a connected undirected graph, the collection of all cycle-free edge sets forms the independence system of a matroid (called a \emph{graphic matroid}), whose bases are exactly the spanning trees of $G$.
It follows from \ref{M:augmentation} that all bases have the same cardinality.
A set $X$ that is not independent (i.e. $X \notin \I$) is called \emph{dependent}.
An inclusion-wise minimal dependent set is called a \emph{circuit}.
A circuit of cardinality $1$ is also called a \emph{loop}.
Two elements $e, f \in E$ are called \emph{parallel} if $\{e, f\}$ is a circuit.
In graphic matroids the terms circuit, loop, and parallel elements align completely with their graph theory counter parts, i.e. cycle, self-loop, and parallel edges, respectively.
The \emph{rank} $\rank(X)$ of a set $X \subseteq E$ is the size of any maximal independent set contained in $X$, i.e. $\rank(X) = \max \{|I| : I \subseteq X, I \in \I\}$.
We may also refer to the rank $\rank(B)$ of any base $B \in \B$, as the \emph{rank of the matroid} $M$, abbreviated with $\rank(M)$.

A \emph{cocircuit} of matroid $M$ is an inclusion-wise minimal set among those sets which have a non-empty intersection with every base of $M$.
Notice that cocircuits correspond to inclusion-wise minimal cuts in the graphic matroid.
A cocircuit of cardinality $1$ (i.e. an element that is contained in every base) is called \emph{coloop}.
We denote by $\C^\ast$ the collection of all cocircuits of matroid $M$.
Cocircuits (and also circuits) of matroids also obey the following cocircuit exchange property.
For all $C^\ast_1, C^\ast_2 \in \C^\ast$ with $e \in C^\ast_1 \cap C^\ast_2$, there exists $C^\ast_3 \in \C^\ast$, such that $C^\ast_3 \subseteq (C^\ast_1 \cup C^\ast_2)-e$ (Proposition~1.4.12 in \cite{oxley2006matroid}). 

Given a matroid $M = (E, \I)$ and a subset $Z \subseteq E$, we can derive two new matroids on ground set $E \setminus Z$ by \emph{deleting} and \emph{contracting} the elements in $Z$.
The deletion operation results in the matroid $M \delete Z \coloneqq (E \setminus Z, \I \delete Z)$ with independent sets $\I \delete Z \coloneqq \{I \subseteq E \setminus Z \mid I \in \I\}$.
The contraction operation results in the matroid $M \contract Z \coloneqq (E \setminus Z, \I \contract Z)$ with independent sets $\I \contract Z \coloneqq \{I \subseteq E \setminus Z \mid I \cup Z^\B \in \I\}$, where $Z^\B$ denotes an arbitrary base of $M \delete (E \setminus Z)$.
It is well-known that each of the two \emph{minors} $M \delete Z$ and $M \contract Z$ is again a matroid (see Section~39.3 in \cite{schrijver2003combinatorial}).
For the sake of convenience, we may also write $\mathcal{S} \delete Z$ and $\mathcal{S} \contract Z$ for $\mathcal{S} \in \{\B, \C, \C^\ast\}$ for the set of bases, circuits, cocircuits of $M \delete Z$ and $M \contract Z$, respectively.
If $Z$ is a singleton $\{e\}$, we may also write $M \delete e$ and $M \contract e$ instead of $M \delete \{e\}$ and $M \contract \{e\}$, respectively.
Also note that a series of deletions and contractions can be performed in any order of the elements, the resulting minor will always be the same (Proposition~3.1.25 in \cite{oxley2006matroid}).
Hence, it is valid to write $M \delete D \contract F$ to denote the minor of $M$ after deleting $D$ and contracting $F$ in any particular order.

Given a matroid $M = (E, \I)$ with base set $\B$, and a weight function $w\colon E \to \mathbb{R}_+$, a base $\hB \in \B$ of maximal weight $w(\hB) = \sum_{e \in \hB} w(e)$ can be computed with the \emph{matroid greedy algorithm} which starts with the initially empty independent set $I \coloneqq \emptyset$, and iteratively, in the order of non-increasing weights $w(e_1) \geq \ldots \geq w(e_m)$, adds the next element $e_t$ to $I$ in iteration $t \in [m]$ if (and only if) $I+e_t \in \I$.
For details see Theorem~40.1 in \cite{schrijver2003combinatorial}.
Hence, with the greedy algorithm, it is easy to determine whether $I$ is a base or whether $E \setminus X$ for some $X$ still contains a base.
Those kind of queries will become important in the auction presented in the next section.

\section{Description of the Ascending Auction}\label{sec:auction}
Recall that we consider an auction in which an auctioneer is selling a finite set of distinguishable, indivisible items $E$ to a finite set of buyers $N$.
We assume that each buyer $i \in N$ is interested in buying only the items in a subset $E_i \subseteq E$, and that each buyer $i \in N$ has an additive valuation function, i.e. they have a non-negative integral valuation $v_i(e) \in \mathbb{Z}_+$ for each item $e \in E_i$ and for each subset $X \subseteq E_i$, we write $v_i(X) = \sum_{e \in X} v_i(e)$.
In the auction we consider, the auctioneer is constrained to sell a base of a given matroid $M = (E, \I)$ on $E$, where $\I$ denotes the collection of independent sets.
We may assume that $E$ only contains those items for which at least one buyer is interested in buying them.
Notice that we can assume w.l.o.g. that no two players are interested in buying the same item.
If not, say if $e \in E_i \cap E_j$ for some item $e \in E$ and two buyers $i \neq j$, we can replace $e$ by two parallel copies $e^i$ and $e^j$, and replace $M = (E, \I)$ by $M' = (E', \I')$ with ground set $E' = E-e \cup \{e^i, e^j\}$ (and $E_i' = E_i-e+e^i, E_j' = E_j-e+e^j$) and independence system $\I'$, where
$\I' = \{I \in \I \mid e \notin I\} \cup \{I-e+e^i, I-e+e^j : I \in \I, e \in I\}$.
This construction extends naturally to the case when even more than two buyers are interested in the same item.

In other words, we can assume that $\{E_i\}_{i \in N}$ is a partition of $E$.
Thus, an instance of our auction is defined by a matroid $M = (E, \I)$ on a finite set $E$, a partition of $E$ into $E = \bigcup_{i \in N} E_i$, and a valuation function $v\colon E \to \mathbb{R}_+$ with $v(e) = v_i(e)$ for the (unique) buyer $i \in N$ with $e \in E_i$.
We furthermore assume that none of the buyers initially is a monopsonist in the sense that no cocircuit of $M$ is contained in any $E_i$.
Otherwise, we could just resolve the monopsonies by selling an item from the monopsony (buyer's choice) for price $0$ in the same way we resolve monopsonies later in the auction (see below).

Before we consider the ascending auction, we remark that there exists a \emph{sealed-bid} Vickrey auction which runs in strongly polynomial time for this model: 
the auctioneer starts by collecting the bids $b(e) \in \mathbb{Z}_+$ from every buyer $i \in N$ for every item $e \in E_i$, which may or may not be equal to her true valuation $v_i(e)$.
Based on the resulting bid vector $b = (b(e))_{e \in E}$, the auctioneer first computes a base $\hB \in \B$ of maximum $b$-weight.
To compute $\hB$, the auctioneer might use the matroid greedy algorithm as described above.
Afterwards, the auctioneer computes for every buyer $i \in N$ the Vickrey price, which incentivizes truthful bidding (cf. \cite{vickrey1961counterspeculation,clarke1971multipart,groves1973incentives}).
In our setting the Vickrey price of buyer $i \in N$ is
\begin{equation}\label{eq:VCG}
	\hp_{\hB}(i) = \max_{B \in \B \delete E_i} b(B) - b(\hB \setminus E_i).
\end{equation}
Since for every $i \in N$ a base of maximum $b$-weight in $M \delete E_i$ can be computed with the matroid greedy algorithm, the auction runs in strongly polynomial time.
As already argued in the introduction, we prefer an ascending auction over a sealed-bid auction, since an ascending auction is more transparent and requires less information from the buyers.

The paper \cite{bikhchandani2011ascending} describes and analyzes an ascending auction (more precisely a \emph{clock auction}) that meets the three mechanism design goals to 
\begin{enumerate*}[label=(\roman*)]
	\item incentivize truthful signals,
	\item compute a welfare-maximizing base, and
	\item run in strongly polynomial time.
\end{enumerate*}
In this section, we give a slightly modified but equivalent presentation of the auction introduced in \cite{bikhchandani2011ascending}.
An analysis of the auction, which is significantly simpler than the one in \cite{bikhchandani2011ascending}, is presented in Section~\ref{sec:analysis}.

\paragraph{Sketch of the ascending auction}
A key ingredient of the auction is that the auctioneer has to keep track of the cocircuits of the currently observed matroid $M \delete D \contract I$ (a minor of the original matroid $M$), where $D$ is the set of items which have been deleted during the auction, and $I$ is the set of already sold items.
As cocircuits have a non-empty intersection with every base, the auctioneer has to make sure that he sells at least one item of each cocircuit in order to sell a base.
The idea now is that the auctioneer increases the price for all items simultaneously.
At any time, a buyer $j$ has the option to announce that there is an item $f \in E_j$ that she would not buy if the price (currently $p$) would increase further (she also can announce this for multiple items of her set and other buyers might also do so at the same time).
In the sequel, we refer to those items as \emph{critical} at price $p$.
Then, the auctioneer removes critical items from the matroid one after another.
This removal will be carried out as a \emph{deletion} in the matroid.
Now as every item occurs in some cocircuit (unless it is a loop) and the cocircuits of the restricted matroid are the same up to removal of the deleted item (see Lemma~\ref{lem:cocircuits-before-deletion} below), some cocircuits decrease in size.
At some point there will be a cocircuit $C^\ast$ that is completely contained in some buyer's item set $E_i$; in that case, we call $C^\ast$ a \emph{monopsony}.\footnote{In \cite{bikhchandani2011ascending}, they refer to such a cocircuit as a \emph{monopoly}, which is accurate in case of a procurement auction. However, we feel that the term monopsony is more natural and consistent with our usage of the terms buyers, auctioneer, etc.}
Whenever a monopsony $C^\ast \subseteq E_i$ occurs, the auctioneer asks the buyer $i$ (the monopsonist) to name her most valuable item $e$ from $C^\ast$.
This announced item $e$ will then immediately be sold to buyer $i$ at price $p$ (in the words of \cite{ausubel2004efficient}, buyer $i$ has \emph{clinched} an item).
After item $e$ is sold, the auctioneer removes $e$ from the matroid.
This removal, however, will be performed as a \emph{contraction}.
After all monopsonies have been resolved (there might be more than one after a deletion), the auctioneer continues deleting critical items (some of those might have become a coloop after deletion of other items at the current price and hence, were sold), selling items whenever a monopsony arises, and increasing prices after all critical items have been removed.

The pseudocode \ref{alg:ascending-auction} below gives a precise description of the auction.
We amended the original pseudocode of \cite{bikhchandani2011ascending} only in the sense that we left out the tie-breaking (which is arbitrary anyway) and perform the deletion steps right away instead of after resolving all monoposonies that occur.

Initially, the auctioneer only knows the matroid $M = (E, \I)$ and the partition $\{E_i\}_{i \in N}$.
For all further information, the auctioneer has to receive \emph{signals} from the buyers, i.e. answers to queries about their valuation functions (see paragraph on truthful signaling below).
As usual in matroid optimization, we assume that the auctioneer has access to an independence oracle.
That is, he might query for any set $X \subseteq E$, whether $X \in \I$ in time $\mathcal{O}(1)$.
Note that, using an independence oracle, the auctioneer can efficiently check whether the conditions in the two \textbf{while}-loops are satisfied.

\begin{algorithm}[H]
	\SetAlgoRefName{Ascending Matroid Auction}
	\caption{$M = (E, \I)$ monopsony-free, partition $\{E_i\}_{i \in N}$ of $E$}
	\label{alg:ascending-auction}
	$I \coloneqq \emptyset, p \coloneqq 0, \hp(e) \coloneqq 0$ for all $e \in E$\\
	\While{$I \notin \B$}{
		$p \coloneqq p+1$\label{alg:ascending-auction:increment}\\
		$R(p) \coloneqq \{f \in E \mid f \text{ announced as critical at price } p\}$\label{alg:ascending-auction:critical}\\
		\For{$f \in R(p)$}{
			$M \coloneqq M \delete f$\\
			\While{$M$ contains a monopsony $C^\ast \subseteq E_i$ of some buyer $i$}{
				$e \coloneqq$ item from $C^\ast$ that buyer $i$ chooses to buy at price $p$\label{alg:ascending-auction:best}\\
				$I \coloneqq I+e$\\ 
				$\hp(e) \coloneqq p$\\
				$M \coloneqq M \contract e$\\
			}
		}
	}
	\Return{$\hB \coloneqq I$ \textup{and} $\hat{p}$}
\end{algorithm}

\paragraph{Truthful signaling}
We emphasize that the communication between auctioneer and buyers takes place only in lines~\ref{alg:ascending-auction:critical} and~\ref{alg:ascending-auction:best}.
We say that buyer $i$ is \emph{signaling truthfully} if her signals to the auctioneer are consistent with $v_i$ in the following sense:
If buyer $i$ is asked to report on her critical items in line~\ref{alg:ascending-auction:critical}, she truly announces those items $e$ as critical for which $v_i(e) = p$.
Furthermore, whenever the auctioneer asks buyer $i$ for her most-valued item in her monopsony $C^\ast \subseteq E_i$ in line~\ref{alg:ascending-auction:best}, she truly selects the item of largest $v$-value in her monopsony $C^\ast$.

\begin{example*}
	Consider the graphic matroid in Figure~\ref{fig:sell-event} with edges labeled by its buyer's valuation.
	At price $p=1$, there are no critical items (Figure~\ref{fig:sell-event:1}).
	Hence, the auctioneer raises the price of all items to $p=2$ at which multiple items become critical.
	After deleting one of those critical items, the auctioneer observes that three buyers now hold a monopsony (Figure~\ref{fig:sell-event:2}).
	While two of those monopsonies are singletons, one buyer still has to signal to the auctioneer which item in her cocircuit she values most (Figure~\ref{fig:sell-event:3}).
	Note that in this example untruthful signaling	would be a strategy that is inconsistent with every possible valuation function since one of the two items was announced as critical while the other one was not).
	After contracting all best elements from each monopsony, the auctioneer is left with two more parallel elements, both of which are critical, so he performs another deletion (Figure~\ref{fig:sell-event:4}) resulting in another singleton monopsony (Figure~\ref{fig:sell-event:5}) which can be contracted.
\end{example*}

\begin{figure}
	\centering
	\subcaptionbox{$p=1$, no critical items, no monopsonies.\label{fig:sell-event:1}}[.3\linewidth]{
		\centering
		\begin{tikzpicture}
			\node[vertex] (1) at (0:1) {};
			\node[vertex] (2) at (72:1) {};
			\node[vertex] (3) at (144:1) {};
			\node[vertex] (4) at (216:1) {};
			\node[vertex] (5) at (288:1) {};
			\draw[player1] (1) -- node[label,above right]{2} (2);
			\draw[player2] (2) -- node[label,above]{2} (3);
			\draw[player2] (2) -- node[label,below right]{3} (4);
			\draw[player3] (3) -- node[label,left]{2} (4);
			\draw[player3] (5) -- node[label,below right]{4} (1);
			\draw[player4] (4) -- node[label,below]{5} (5);
		\end{tikzpicture}
	}
	\subcaptionbox{$p=2$, three critical items, delete the one by buyer~1, three monopsonies.\label{fig:sell-event:2}}[.3\linewidth]{
		\centering
		\begin{tikzpicture}
			\node[vertex] (1) at (0:1) {};
			\node[vertex] (2) at (72:1) {};
			\node[vertex] (3) at (144:1) {};
			\node[vertex] (4) at (216:1) {};
			\node[vertex] (5) at (288:1) {};
			\draw[player2] (2) -- node[label,above]{2} (3);
			\draw[player2] (2) -- node[label,below right]{3} (4);
			\draw[player3] (3) -- node[label,left]{2} (4);
			\draw[player3] (5) -- node[label,below right]{4} (1);
			\draw[player4] (4) -- node[label,below]{5} (5);
			\draw[monopsony=rwthred] (0, 1.3) -- (0, .5) -- (.8, .5);
			\draw[monopsony=rwthblue] (1.2, .3) -- (.8, .3) -- (.8, -.3) -- (1.2, -.3);
			\draw[monopsony=rwthgreen] (-.8, -1) -- (1, .55);
		\end{tikzpicture}
	}
	\subcaptionbox{State after contracting items for buyer~3 and~4, buyer~2 decides to take her higher value item for price $2$.\label{fig:sell-event:3}}[.3\linewidth]{
		\centering
		\begin{tikzpicture}
			\node[hidden] (1) at (0:1) {};
			\node[vertex] (2) at (72:1) {};
			\node[vertex] (3) at (144:1) {};
			\node[vertex] (4) at (216:1) {};
			\node[hidden] (5) at (288:1) {};
			\draw[player2] (2) -- node[label,above]{2} (3);
			\draw[player2] (2) -- node[label,below right]{3} (4);
			\draw[player3] (3) -- node[label,left]{2} (4);
			\draw[monopsony=rwthred] (0, 1.3) -- (0, .5) -- (.8, .5);
		\end{tikzpicture}
	}
	\subcaptionbox{Delete critical item of buyer~2.\label{fig:sell-event:4}}[.3\linewidth]{
		\centering
		\begin{tikzpicture}
			\node[hidden] (1) at (0:1) {};
			\node[vertex] (2) at (72:1) {};
			\node[vertex] (3) at (144:1) {};
			\node[hidden] (4) at (216:1) {};
			\node[hidden] (5) at (288:1) {};
			\draw[player2] (2) to[bend right] node[label,above]{2} (3);
			\draw[player3] (3) to[bend right] node[label,below]{2} (2);
		\end{tikzpicture}
	}
	\subcaptionbox{Buyer~3 has now a monopsony, sell her the item.\label{fig:sell-event:5}}[.3\linewidth]{
		\centering
		\begin{tikzpicture}
			\node[hidden] (1) at (0:1) {};
			\node[vertex] (2) at (72:1) {};
			\node[vertex] (3) at (144:1) {};
			\node[hidden] (4) at (216:1) {};
			\node[hidden] (5) at (288:1) {};
			\draw[player3] (3) to[bend right] node[label,below]{2} (2);
			\draw[monopsony=rwthblue] (.05, 1.2) -- (.05, .7) -- (.6, .6);
		\end{tikzpicture}
	}
	\caption{
		Example iteration on a graphic matroid with four buyers 1 (black, solid), 2 (red, dashed), 3 (blue, dotted), 4 (green, dotdashed): For $p=0$ and $p=1$ there are no critical items.
		When $p$ increases to $2$ multiple items get critical and trigger some sales, including some of critical items.
	}
	\label{fig:sell-event}
\end{figure}
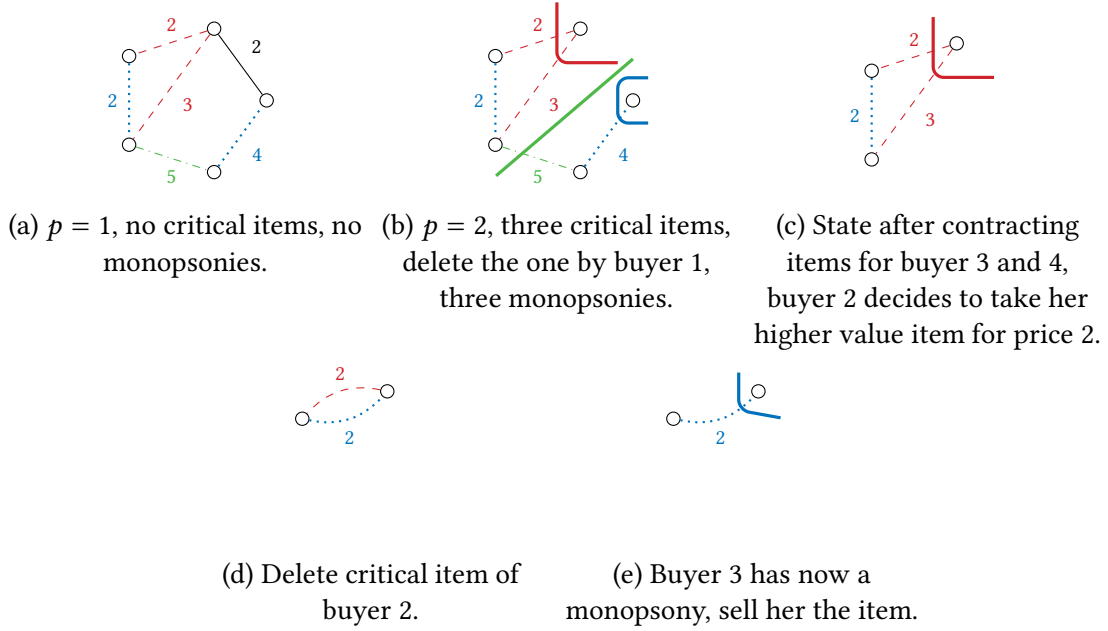

\section{Analysis of the Auction}\label{sec:analysis}
Recall the three mechanism design goals \ref{goals:eq}-\ref{goals:poly} described in the Introduction.
Computational efficiency, i.e. property~\ref{goals:poly}, is clearly fulfilled:
The running time of the auction scales polynomially in the number of items and buyers.
It even runs in strongly polynomial time by using a long-step variant that only requires a reasonable consistency check on the signals (see \cite{bikhchandani2011ascending} and \cite{raach2023implementing} for the polymatroid case).
That is, by replacing line~\ref{alg:ascending-auction:increment} in \ref{alg:ascending-auction} with 

{
\RestyleAlgo{plain}
\begin{algorithm}[H]
	\setcounter{AlgoLine}{2}
	\SetNlSty{}{}{'}
	~\textbf{\textbar}~~~~$p \coloneqq \min_{i \in N} \{q_i \mid \text{smallest } q_i > p \text{ for which } i \text{ has a critical item}\}$\label{alg:ascending-auction:incrementLong}
\end{algorithm}
}

and the understanding that in line~\ref{alg:ascending-auction:critical} every buyer who minimized the term in line~\ref{alg:ascending-auction:incrementLong} has to announce at least one item as critical, we obtain a long-step auction that results in the same outcome.

The keys to show that the \ref{alg:ascending-auction} yields a welfare maximizing allocation (under truthful signals) and assigns Vickrey prices are a few statements about cocircuits and how they evolve by taking minors of the matroid.

In essence, we will need the following three matroid folklore lemmas.
\begin{restatable}[Cocircuits after deletion]{lemma}{cocircuitsAfterDeletion}\label{lem:cocircuits-after-deletion}
	Let $C^\ast$ be a cocircuit of $M$ and $e \in E$.
	Then $C^\ast-e$ is a union of cocircuits of $M \delete e$.
\end{restatable}
\begin{restatable}[Cocircuits before deletion]{lemma}{cocircuitsBeforeDeletion}\label{lem:cocircuits-before-deletion}
	Let $C^\ast$ be a cocircuit of $M \delete e$.
	Then $C^\ast$ or $C^\ast+e$ is a cocircuit of $M$.
\end{restatable}
\begin{restatable}[Augmentation by cocircuit, \cite{dawson1980optimal}]{lemma}{augmentationByCocircuits}\label{lem:augmentation-by-cocircuits}
	Let $I$ be an independent set of a matroid $M$ that can be extended to a maximum weight base.
	Let $C^\ast$ be any cocircuit of $M \contract I$ and $e \in C^\ast$ its maximum weight element.
	Then $I+e$ is independent and can be extended to a maximum weight base.
\end{restatable}
Since the proofs of these lemmas are not always easy to find (including Dawson's Lemma, which we do not state verbatim above), we added their proofs in Appendix~\ref{apx:folklore}.

Now let us state and prove the main theorem of \cite{bikhchandani2011ascending}.

\begin{theorem}[Theorem~13 in \cite{bikhchandani2011ascending}]\label{thm:main-theorem}
	The \ref{alg:ascending-auction} returns a welfare-maximizing allocation and assigns Vickrey prices.
\end{theorem}
\begin{proof}
	For the proof, we assume truthful signals from the buyers.
	The discussion in Section~\ref{sec:discussion} shows why this is justified.
	
	First, we show that the resulting allocation $\hB$ is welfare-maximizing.
	For that it suffices to show that $I$, the set of sold items in the \ref{alg:ascending-auction}, is max-weight base extendable throughout the whole auction.
	Since we start with $I = \emptyset$, the statement is obviously true in the beginning.
	Given some stage of the auction with sold items $I$ and deleted items $D$, we consider the minor $M \delete D \contract I$.
	Say $e \in E \setminus (D \cup I)$ is the next item to be added to $I$.
	This means $e$ is an item of maximum value in some monopsony $C^\ast$, i.e. a cocircuit of $M \delete D \contract I$.
	Thus, by applying Lemma~\ref{lem:cocircuits-before-deletion} repeatedly, we have that $\tilde{C}^\ast = C^\ast \cup D'$ is a cocircuit of $M \contract I$ for some $D' \subseteq D$.
	As $v(e) \geq v(f)$ for all $f \in D$ (and in particular $D'$), $e$ has also maximum value in $\tilde{C}^\ast$.
	Hence, by Lemma~\ref{lem:augmentation-by-cocircuits}, $I+e$ is max-weight base extendable.

	Now let us show that the auction also charges Vickrey prices to the buyers.
	For each $e \in \hB$, let $f_e$ denote the item which is deleted before $e$ was added to $I$, i.e. $f_e$ is the price-setting item for $e$ and we have $p(e) = v(f_e)$ .
	We need to show that for all $i \in N$
	\[
		\hB' = \hB \setminus E_i \cup \{f_e : e \in \hB \cap E_i\}
	\]
	is a max-weight base of $M \delete E_i$.
	Note that a deletion of an item $f$ can create at most one monopsony per buyer: 
	If $C^\ast_1, C^\ast_2$ are both monopsonies of buyer $j$, then $C^\ast_1+f$ and $C^\ast_2+f$ are both cocircuits in the matroid before deleting $f$ and hence, by the cocircuit exchange property, there exists a cocircuit $C^\ast_3 \subseteq (C^\ast_1 \cup C^\ast_2)-f$ which would be a monopsony of $j$ before deleting $f$.
	That is, $\hB'$ has the correct cardinality to be a base of $M \delete E_i$ (recall that initially there is no monopsony in $M$).
	Thus, we only need to show that $\hB'$ is indeed welfare-maximizing and independent $M \delete E_i$.
	We again use Lemma~\ref{lem:augmentation-by-cocircuits} to show that all elements $e \in \hB \setminus E_i$ are still suitable for the new base and that for all $e \in \hB \cap E_i$, its corresponding $f_e$ is a suitable replacement.
	Consider some $e \in \hB$ and the set $I$ before $e$ was added.
	Then $e$ was an item of maximum value in some monopsony $C^\ast$ of $M \delete D \contract I$ .
	We also have $v(e) \geq v(f)$ for all $f \in D$ .
	Thus, $e$ also has maximum value in $\tilde{C}^\ast = C^\ast \cup D'$, a cocircuit of $M \contract I$ with $D' \subseteq D$ by applying Lemma~\ref{lem:cocircuits-before-deletion}.
	Moreover $f_e \in D'$ as otherwise $C^\ast$ would have been a monopsony before $f_e$ was deleted.
	From here, we make a case distinction:

	\textbf{Case 1:} $e \notin E_i$.
	By applying Lemma~\ref{lem:cocircuits-after-deletion} on $M \contract I$ w.r.t. all the items in $E_i$, we also have that $e$ has maximum value in some cocircuit of $M \contract I \delete E_i = M \delete E_i \contract I$ and hence, by Lemma~\ref{lem:augmentation-by-cocircuits}, $I+e$ is max-weight base extendable for $M \delete E_i$.

	\textbf{Case 2:} $e \in E_i$.
	Deleting $E_i$ (including $e$) yields the minor $M \contract I \delete E_i = M \delete E_i \contract I$.
	Note that $v(f_e) \geq v(f)$ for all $f \in D$ as $f_e$ was deleted last.
	Applying Lemma~\ref{lem:cocircuits-after-deletion} on $M \contract I$ w.r.t. all items in $E_i$ yields that there is a cocircuit $\breve{C}^\ast \subseteq \tilde{C}^\ast$ of $M \delete E_i \contract I$ with $f_e$ as item of maximum value.
	So by Lemma~\ref{lem:augmentation-by-cocircuits}, we have that $I+f_e$ is max-weight base extendable in $M \delete E_i$.	
\end{proof}

\subsection{From Theorem~\ref{thm:main-theorem} to Truthful Equilibrium}\label{sec:discussion}
Note that one cannot simply apply the VCG Theorem \cite{vickrey1961counterspeculation,clarke1971multipart,groves1973incentives} to claim that the auction incentivizes truthful signaling just because the prices match Vickrey prices.
The reason is that, in contrast to a sealed-bid auction, in an ascending auction buyers gain information about the other buyers' valuations as the auction progresses.
Indeed, in the example in Appendix~\ref{apx:not-dominant} it is shown that truthful signaling is not a dominant strategy in this auction (and in fact, no dominant strategy need to exist).
However, the following theorem can be shown.
\begin{theorem}[Theorem~17 in \cite{bikhchandani2011ascending}]\label{thm:incentives}
	Truthful signaling is an ex-post equilibrium of the \ref{alg:ascending-auction}.
\end{theorem}
The proof is a bit tedious but straight-forward:
A signaling strategy that is inconsistent with any valuation function can always be replaced with a consistent signaling strategy without changing the outcome for any buyer, provided that all other buyers signal truthfully.
However, among all consistent strategies, a truthful one is best by the VCG Theorem.
We refer to \cite{bikhchandani2011ascending} for proof details.

This justifies that (here and in \cite{bikhchandani2011ascending}) Theorem~\ref{thm:main-theorem} is proven under the assumption that all buyers are signaling truthfully.

\bibliographystyle{alpha}
\bibliography{references}

\appendix
\section{Proofs for Matroid Folklore Lemmas}\label{apx:folklore}
Recall that our simplified proof of Theorem~\ref{thm:main-theorem} relies on the three folkloric facts about matroids, namely Lemma~\ref{lem:cocircuits-after-deletion}--\ref{lem:augmentation-by-cocircuits}.
We call these lemmas folkloric since even to an expert on matroid theory, their origin might very well be unknown, even though their statements immediately ring true in the ear of a person who is just a little bit familiar with matroids.
In this appendix, we give proofs to all three lemmas.

Lemma~\ref{lem:cocircuits-after-deletion} is a special case of the dual statement of an exercise in \cite{oxley2006matroid}.

\cocircuitsAfterDeletion*
\begin{proof}
	We prove the following stronger statement (see \cite{oxley2006matroid}, Exercise~2 in Section~3.1):
	
	\textit{
		Let $C \in \C$ and $e \in E$.
		\begin{enumerate}[label=\textit{(\alph*)}]
			\item If $e \in C$, then $\{e\} \in \C$ or $C-e \in \C(M \contract e)$.
			\item If $e \notin C$, then $C$ is a union of circuits of $\C(M \contract e)$.
		\end{enumerate}
	}
	
	Our statement follows as the dual statement by taking deletions in the cocircuit set instead of contractions in the circuit set and the understanding that the empty set is an empty union of cocircuits if $e$ was a coloop and a single cocircuit is the union over a singleton. 
	\begin{enumerate}[label=\textit{(\alph*)}]
		\item Let $e \in C \in \C$.
			If $C = \{e\}$, there is nothing to show.
			Otherwise, consider $M \contract e$.
			Assume $C-e \notin \C \contract e$, i.e. $C-e \in \I \contract e$ or $C-e$ properly contains a circuit of $M \contract e$.
			As $\I \contract e = \{I \subseteq E-e \mid I+e \in \I\}$, the former cannot be true, so assume the latter case.
			Let $C' \subset C-e$ with $C' \in \C \contract e$.
			But then $C'+e \subset C$ and hence, $C'+e \in \I$ thus, $C' \in \I \contract e$.
		\item Let $e \notin C \in \C$.
			If $C \in \C \contract e$, we are done.
			Otherwise, note that $C \notin \I \contract e$ as the contrary would imply $C+e \in \I$ and hence, $C \in \I$.
			To show that $C$ does not contain any $f$ outside of a circuit of $M \contract e$ with other items from $C$, we show that
			\[
				\rank_{M \contract e}(C) = \rank_{M \contract e}(C-f)
			\]
			for all $f \in C$.
			Suppose not, i.e. $\rank_{M \contract e}(C) = \rank_{M \contract e}(C-f) + 1$ for some $f \in C$.
			Then using the identity from Proposition~3.1.6 in \cite{oxley2006matroid}
			\begin{equation*}
				\rank_{M \contract Z}(X) = \rank(X \cup Z) - \rank(Z)
			\end{equation*}
			we also have that
			\begin{alignat*}{3}
				&& \rank(C+e) - \rank(e) &= \rank(C-f+e) - \rank(e) + 1\\
				\iff\ && \rank(C+e) &= \rank(C-f+e) + 1,
			\end{alignat*}
			which implies with submodularity and $|\{f\}| = 1$ that $\rank(C) = \rank(C-f) + 1$ and hence, $C \notin \C$.\qedhere
	\end{enumerate}
\end{proof}

The second folklore fact is mentioned (in its dual version) at the end of Section~39.3 in \cite{schrijver2003combinatorial} but without proof.

\cocircuitsBeforeDeletion*
\begin{proof}
	Clearly, at most one can be true.
	As $C^\ast \in \C^\ast(M \delete e)$, we know that $C^\ast \cap B \neq \emptyset$ by definition for all $B \in \B(M \delete e)$.
	Thus, $E \setminus C^\ast$ does not contain a base of $M \delete e$.
	Also $E \setminus (C^\ast+e)$ does not contain a base of $M$:
	Suppose it would, i.e. $B \subseteq E \setminus (C^\ast+e)$ with $B \in \B$, then $B \in \I$ and hence, $B \in \I(M \delete e)$ as $e \notin B$.
	Now assume $E \setminus (C^\ast \setminus T)$ does not contain a base $B$ for some maximal $\emptyset \subset T \subseteq C^\ast$, making $\tilde{C}^\ast \coloneqq C^\ast \setminus T \subset C^\ast$ a cocircuit in $M \delete e$.
	But then $\tilde{C}^\ast \in \C^\ast(M \delete e)$, which contradicts the fact that $C^\ast \in \C^\ast(M \delete e)$.
\end{proof}

Finally, we give a short proof of Dawson's Lemma \cite{dawson1980optimal} using a well-known theorem by Brualdi \cite{brualdi1969comments}.

\augmentationByCocircuits*
\begin{proof}
	Let $\hB$ be a base of maximum weight with $I \subseteq \hB$ with respect to weight function $w$.
	If $e \in \hB$, we are already done.
	Otherwise, choose $B_0 \in \B \contract I$ with $B_0 \cap C^\ast = \{e\}$.
	Such a $B_0$ always exists as $C^\ast$ is by definition an inclusion-wise minimal set in $E \setminus I$ that intersects every base of $M \contract I$.
	Now construct a base of $M$ by taking $B = B_0 \cup I$ and observe that $B \cap C^\ast = \{e\}$ holds as well (as $C^\ast \subseteq E \setminus I$), which means that $C^\ast$ is also a cocircuit of $M$.
	Since $e \in B \setminus \hB$, there exists some $f \in \hB \setminus B$ such that both, $\hB-f+e \in \B$, and $B-e+f \in \B$ (this is the well-known \emph{strong base exchange} property of matroids as introduced by \cite{brualdi1969comments}).
	Note that $f \notin I$, since $f \in \hB \setminus B \subseteq \hB \setminus I$.
	Then it also follows that $f \in C^\ast$, as otherwise $(B_0-e+f) \cap C^\ast \subseteq (B-e+f) \cap C^\ast = \emptyset$, contradicting the fact that $C^\ast$ is a cocircuit of $M$ and $M \contract I$.
	Since $\hB$ is a base of maximum weight, the weight of base $\hB-f+e$ cannot be larger; thus, $w(e) \leq w(f)$.
	However, since $e, f \in C^\ast$ and $e$ is max-weight element of $C^\ast$, we have $w(e) = w(f)$), implying that $\hB-f+e$ is a max-weight base containing $I+e$.
\end{proof}

\section{Discussion on Incentives}\label{apx:not-dominant}
After reading Theorem~\ref{thm:main-theorem}, one might think that together with the VCG Theorem it is implied that the presented auction is indeed \emph{strategy-proof}, i.e. that telling the auctioneer the truth at all times is a dominant strategy.
Before we argue that this is in fact not the case, let us first look at a sealed-bid proxy auction (cf. the Revelation Principle \cite{myerson1981optimal}):
at the beginning of the auction, all buyers submit a bid $b_i\colon E_i \to \mathbb{R}_+$ to a proxy buyer (e.g. the auctioneer himself), who does not know whether $b_i$ is truthful or not.
Then the proxy buyers answer all the requests of the auctioneer in Lines~\ref{alg:ascending-auction:critical} and~\ref{alg:ascending-auction:best} of the \ref{alg:ascending-auction} consistently with $b_i$.
As we have shown, the \ref{alg:ascending-auction} will now return an optimal allocation (w.r.t. the $b_i$s) and Vickrey prices.
By the VCG Theorem, the dominant strategy for each buyer is to submit her true valuations to her proxy buyer.
The proxy auction is sketched in Figure~\ref{fig:proxy-auction}.

\begin{figure}
	\centering
	\begin{tikzpicture}[>=latex']
		\draw[-] (0, 0) rectangle (5, 3);
		\draw[-] (2, .5) rectangle node[align=center,text width=2cm]{\small\ref{alg:ascending-auction}} (4.5, 2.5);
		
		\node[draw,circle,inner sep=0pt,minimum size=5pt] (p1) at (.5, 2.75) {};
		\node[draw,circle,inner sep=0pt,minimum size=5pt] (p2) at (.5, 2.25) {};
		\node at (.5, 1.25) {$\vdots$};
		\node[draw,circle,inner sep=0pt,minimum size=5pt] (pn) at (.5, .25) {};

		\draw[->] (-.5, 2.75) node[left]{$b_1$} -- (p1);
		\draw[->] (-.5, 2.25) node[left]{$b_2$} -- (p2);
		\node at (-.75, 1.25) {$\vdots$};
		\draw[->] (-.5, .25) node[left]{$b_n$} -- (pn);

		\draw[->] (p1) to[bend left=15] (2, 2.25);
		\draw[->] (2, 2.25) to[bend left=15] (p1);
		\draw[->] (p2) to[bend left=15] (2, 1.9);
		\draw[->] (2, 1.9) to[bend left=15] (p2);
		\draw[->] (pn) to[bend left=15] (2, .75);
		\draw[->] (2, .75) to[bend left=15] (pn);

		\draw[->] (4.5, 2) -- (5.5, 2) node[right]{$\hB$};
		\draw[->] (4.5, 1) -- (5.5, 1) node[right]{$\hp$};
	\end{tikzpicture}
	\caption{Sketch of the proxy auction.}
	\label{fig:proxy-auction}
\end{figure}
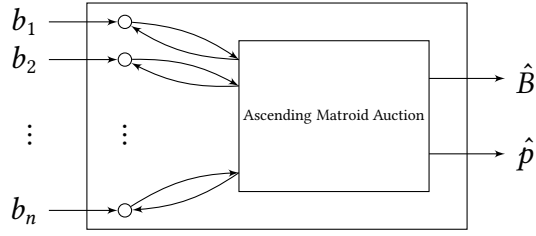

The key difference between the ascending auction and the described proxy auction is that in the latter the buyers are submitting a valuation (truthful or not) and stick with it for the entire auction, while in the former they can in principle answer requests of the auctioneer in a way that is inconsistent in the sense that the signals cannot be represented by any valuation function.
This kind of strategic behavior could be beneficial to a buyer because she can react to events during the auction.
For an example of such inconsistent signaling, see Figure~\ref{fig:inconsistent-signal}.

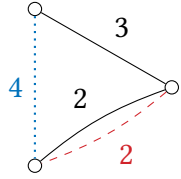
\begin{figure}
	\centering
	\begin{tikzpicture}
		\node[vertex] (1) at (0:1.2) {};
		\node[vertex] (2) at (120:1.2) {};
		\node[vertex] (3) at (240:1.2) {};
		\draw[player1] (1) -- node[above right]{$3$} (2);
		\draw[player1] (1) to[bend right=10] node[above left]{$2$} (3);
		\draw[player2] (1) to[bend left=10] node[below right]{$2$} (3);
		\draw[player3] (2) -- node[left]{$4$} (3);
	\end{tikzpicture}
	\caption{
		At price $p = 2$ two items are (truthfully announced) as critical: 
		one by buyer~1 (black, solid), one by buyer~2 (red, dashed).
		If buyer~2's item gets deleted first by the auctioneer, player~1 has a monopsony and could now announce her critical item as best, which would be inconsistent with the information that her other item is not critical.
	} 
	\label{fig:inconsistent-signal}
\end{figure}

The following example illustrates that a dominant strategy might not exist.
\begin{example*}
	Consider the graphic matroid depicted in Figure~\ref{fig:example-dominance} and two buyers $N = \{1, 2\}$ where $E_i = \{e^i, f^i\}$ for both $i \in N$.
	Let $v_1 = (2, 4)$ and $v_2 = (3, 3)$ (where $v_i = (v_i(e^i), v_i(f^i))$).
	Suppose buyer~2 always signals truthfully (call this strategy $\sigma_2$), then buyer~1 also could signal truthfully by the VCG Theorem (call this strategy $\sigma_1$).
	However, there is a second strategy, let us call it $\sigma_1'$, that yields the same maximal payoff, which is signaling at price $p=1$ already that $e_1$ is critical.
	Note that these two strategies are the only best responses of buyer~1 against truthful signals of buyer~2, thus, if there is a dominant strategy it has to be $\sigma_1$ or $\sigma_1'$.
	On the other hand, buyer~2 could incentivize $\sigma_1'$ by in turn declaring the item $f_2$ as critical at $p=2$ if buyer~1 plays according to $\sigma_1'$ (call this strategy $\sigma_2'$).
	Hence, $\sigma_1$ cannot be dominant as it is not a best response against $\sigma_2'$.
	But buyer~2 could also act spiteful towards buyer~1 by instead only announcing item $f_2$ as critical at $p=4$ if buyer~1 plays strategy $\sigma_1'$ but playing truthful otherwise (call this strategy $\sigma_2''$).
	Then $\sigma_1$ is a better response against $\sigma_2''$ than $\sigma_1'$, so $\sigma_1'$ is not dominant either.

	In essence, the problem for buyer~1 is that she does does not know for her first decision point whether buyer~2's reaction will be altruistic or spiteful towards her decision, which makes neither of her possible strategies dominant.
\end{example*}
\begin{figure}
	\centering
	\begin{tikzpicture}
		\node[vertex] (1) at (0, 0) {};
		\node[vertex] (2) at (3, 0) {};
		\node[vertex] (3) at (6, 0) {};
		\draw[player2] (1) to[bend left=40] node[above]{$2$} node[below]{$e^1$} (2);
		\draw[player2] (2) to[bend left=40] node[above]{$4$} node[below]{$f^1$} (3);
		\draw[player3] (1) to[bend right=40] node[below]{$3$} node[above]{$e^2$} (2);
		\draw[player3] (2) to[bend right=40] node[below]{$3$} node[above]{$f^2$} (3);
	\end{tikzpicture}
	\caption{Graphic matroid with two buyers 1 (red, dashed) and 2 (blue, dotted) where signaling truthfully is not a dominant strategy for buyer~1.}
	\label{fig:example-dominance}
\end{figure}
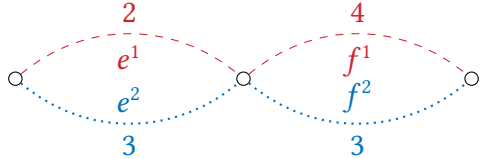

\section{Transparency, Privacy, and Communication Cost}\label{apx:tpc}
As mentioned in the introduction, the adventages of ascending auction---generally speaking---are that the process is more transparent, buyers only need to provide the information that is actually needed to compute the outcome, and the total number of bits that has to be communicated might be lower.
We want to briefly comment on whether the \ref{alg:ascending-auction} fulfills these promises.

In a sealed-bid auction, the auctioneer is the only agent that has full information.
If he sells an item $e$ to a buyer at price $p$, then there should be another buyer who dropped out at price $p$ to justify the price.
It is easy to entertain the idea that an auctioneer with the secret objective of revenue maximization simply can introduce a shill bid to drive up the price after he determined the highest bidder per cocircuit (or just claim that there was another buyer bidding just below the highest bid).
On the other hand, in the \ref{alg:ascending-auction} we can assume that the communication is public, so every buyer can observe by herself whether she is a monopsonist or not.
There is no opportunity for the auctioneer to introduce a shill without risking (we assume that the auctioneer has zero knowledge about the buyers' valuations) that this shill bidder might win an item which would yield zero revenue.
Recently Komo et al. \cite{komo2024shill} investigated which auctions are \emph{shill-proof}, i.e. which auctions keep the auctioneer honest (instead of the buyers, when we discuss strategy-proofness).

Regarding privacy, Milgrom and Segal \cite{milgrom2020clock} adapted the privacy model from social choice theory by Brandt and Sandholm \cite{brandt2005unconditional} for auctions.
In sealed-bid auctions, the auctioneer would just examine all bids to determine allocation and prices (in our case, determine the highest and second-highest bid per cocircuit).
However, this is not really necessary; the auctioneer needs to know the second highest bid in a cocircuit to determine the price but not really the amount of the highest bid.
The information that there is one that is higher than the second highest bid would suffice.
This is what the \ref{alg:ascending-auction} exploits: assuming no ties (for the sake of this argument), the buyer with the highest valuation in a cocircuit will not reveal to the auctioneer what exactly her bid is.
By not reporting any critical item, she just signals to the auctioneer that her bid is still higher than the current price and hence, any potential second-highest bid with that amount.
This property is called \emph{unconditional winner's privacy} in \cite{milgrom2020clock}.

Finally, let us discuss communication cost.
If every buyer has a valuation for every subset of items that she might receive, then in a sealed-bid auction, every buyer would need to communicate an exponential number of values to the auctioneer.
These values do not have to be send expicitly in an ascending auction, they are rather sent implicitly by the signals of the ascending auction, which suffices to compute the outcome.
However, in our model we consider additive valuations and then this argument does not go through (at least not in the worst case) since the submission of an additive valuation only requires $m \log V$ many bits, where $m$ is the number of items and $V$ the maximal valuation over any item.
One can now easily envision a setting where in an ascending auction (with unit increments) the auctioneer alone has to send at least one bit per iteration, which would require a total number of bits that is at least linear in $V$.

\end{document}